\documentclass[12pt]{article}
\usepackage{latexsym}
\usepackage{graphicx}
\usepackage{epsfig}
\usepackage{amsmath, amssymb}
\usepackage{mathrsfs}
\usepackage{natbib}
\usepackage{hyperref}

\renewcommand{\cite}[1]{\citet{#1}}

\newtheorem{theorem}{Theorem}[section]

\newtheorem{corollary}[theorem]{Corollary}

\newtheorem{definition}[theorem]{Definition}

\newtheorem{problem}[theorem]{Problem}

\newcommand{\cX}{{\cal X}}

\newcommand{\cF}{{\cal F}}

\newcommand{\qed}{\mbox{ }~\hfill~$\Box$ \vspace{1ex} }

\newenvironment{proof}{\noindent{\sc Proof: }}{ \qed }

\newcommand{\rmII}{\text{\it I\kern-.08em I\,}}
\newcommand{\rmIII}{\text{\it I\kern-.08em I\kern-.08em I\,}}
\newcommand{\rmIV}{\text{\it I\kern-.08em V\,}}

\newcommand{\bilinear}[2]{{\langle #1 , #2 \rangle}}

\begin{document}

\title{Finance without Probabilistic Prior Assumptions}
\author{Frank Riedel\thanks{ Financial Support through the German Research
Foundation, International Graduate College ``Stochastics and Real World Models'', Research Training Group EBIM, ``Economic Behavior and Interaction Models'', and Grant Ri--1128--4--1
  is gratefully acknowledged.}
\\
Institute of Mathematical Economics\\
Bielefeld University
}
\date{\today\\
$\,$\hfill\\
{\small\sc{This paper is dedicated to the students of the class ``Advanced Topics in Finance'' at Bielefeld University in Summer 2011}}}

\maketitle

\begin{abstract}
\noindent
We develop the fundamental theorem of asset pricing in a probability--free infinite--dimensional setup. We replace the usual assumption of a prior probability by a certain continuity property in the state variable.  Probabilities enter then endogenously as \emph{full support} martingale measures (instead of \emph{equivalent} martingale measures). A variant of the Harrison--Kreps--Theorem on viability and no arbitrage is shown. Finally, we show how to embed the superhedging problem in a classical infinite--dimensional linear programming problem.
\end{abstract}

{\footnotesize{ \it Key words and phrases:} Probability--Free Finance, Fundamental Theorem of Asset Pricing, Full--Support Martingale Measure, Superhedging, Infinite--Dimensional Linear Programming \\
{\it \hspace*{0.6cm} AMS subject classification:  91G99, 90C34, 90C48}\\
{\it \hspace*{0.6cm} JEL subject classification: G12, D53} }

\newpage

\section{Introduction}

In finite state spaces, the basic hedging and pricing results of Mathematical Finance can be derived with simple algebra. In the  binomial tree model of \cite{Coxetal79}, simple linear equations suffice to derive the conditions for no arbitrage, the  equivalent martingale measure, the hedging portfolio, and, in turn, the unique no-arbitrage price. More generally, the Fundamental Theorem of Asset Pricing and the existence of equivalent martingale measures are easily developed by linear algebra even in incomplete market models.

As soon as one moves to more complex, infinite models, or continuous time, one starts by fixing a prior probability $P$ on the state space. From a modeling point of view, this seems natural if one has the picture of an asset price as a random variable in mind. It is also very useful because it allows us to use the powerful methods of probability theory. From a conceptional point of view, one might nevertheless ask: what is the role of the prior probability $P$, and do we really need it to develop Mathematical Finance?

This question might seem a purely technical one at first sight. However, in the light of the recent discussions on Knightian uncertainty, it enters center stage. If the modeler, or the investor, cannot be certain about the choice of the prior, this can result in sensitive and unreliable prescriptions for pricing and hedging of derivatives.

On the other hand, one might reply that the measure $P$ does not really matter for pricing and hedging as the focus is on pricing or martingale measures only. These martingale measures have to be \emph{equivalent} to the original $P$, however, and thus the choice of $P$ plays an implicit role by determining the sets of measure zero, or in other words, the black swan events that get price zero (even if they should not).

The economic role of $P$ has also been  pointed out in  \cite{HarrisonKreps79}; these authors relate the notions of no arbitrage and viability with economic equilibrium:
 \begin{quotation}
  [The probability measure $P$ ] $\ldots$ determines the space $X \, [=L^2(\Omega,\cF,P)]$ of contingent claims, it determines the continuity requirement for [the preference relation] $\succeq \in A$, and \emph{through its null sets} it plays a role in the requirement that $\succeq \in A$ be strictly increasing. (Emphasis by me)
\end{quotation}

Recently,  the attempt to model uncertain volatility has led to very interesting models where one cannot fix the null sets, and where it is not even clear whether it is rational to require knowledge of all null sets. The theory of risk measures and the decision theoretic literature on Knightian ambiguity lead to multiple prior models.
When volatility is unknown in continuous time, one has to work with mutually singular measures, and one cannot fix the null sets in advance (see \cite{Peng07} for the calculus, \cite{Vorbrink10} for no arbitrage considerations, and \cite{EpsteinJi11} for the basic economics and asset pricing consequences).

I therefore think that it is useful to develop the foundations of Finance in a prior--free model.

In this note, I develop the fundamental theorem of asset pricing without any probabilistic assumptions on the general measurable state space. Instead,  a topological continuity requirement is used: the asset payoffs are continuous in the state variable. For modeling purposes, this is merely a technical assumption that is usually satisfied (as one can construct models where the asset payoffs are projections on some product space, and hence continuous)\footnote{In general, I think that one cannot dispense with some topological requirement; e.g., it seems impossible to work with the space of bounded, measurable functions only as there are no strictly positive linear functionals on that space.}.

As one cannot speak of \emph{equivalent} martingale measures in this context, another notion of pricing measure is needed. It turns out that the right concept is \emph{full support} martingale measure. The pricing functional has to assign positive prices to all open sets in order to avoid arbitrage opportunities.
Hence, probabilistic models do enter again, as a consequence of the no arbitrage assumption, but without imposing any probabilistic assumptions ex ante.

With the help of the new version of the fundamental theorem, one can easily characterize all arbitrage--free prices for derivatives as expectations under some full support martingale measure.

As it is important to understand the relations between Finance and Economics, we also give a variant of the Harrison--Kreps theorem which relates the economic viability of asset prices to the possibility to extend the price functional from the marketed space to the whole space of possible contingent claims while keeping strict positivity of the pricing functional.

As an application of the new setup, I show how to embed the superhedging problem in incomplete markets into the language of linear programming in infinite--dimensional spaces. The fact that superhedging and the upper bound on all no arbitrage prices are dual is well known, of course. In the probabilistic setting, it is not straightforward to embed the problems directly into the language of linear programming. This is why usually a combination of stochastic and optimization techniques is used in proving the duality. In our setup, the relation to linear programming is direct and obvious. To the best of my knowledge, this embedding and treatment of the superhedging problem has not been done before.

I am not aware of probability--free approaches to the foundations of Finance in general, in particular for the Fundamental Theorem. Students of Hans F\"{o}llmer\footnote{including the author of this note and Walter Willinger.} are well aware that the basic hedging argument in complete markets can be carried out   without imposing probabilistic prior assumptions even in continuous time; one uses F\"{o}llmer's  pathwise approach to It\^{o} integration (\citeyear{Foellmer81}) to develop the Black--Scholes--like hedging argument in a purely analytic way, compare  \cite{BickWillinger94}. \cite{Foellmer99} develops these ideas and their relations to economics masterfully; see also \cite{Foellmer01} for a similar account in English, and the textbook by \cite{Sondermann06}. Here, our focus is on the Fundamental Theorem and the relation to economic equilibrium in the spirit of Harrison and Kreps.

The paper is set up as follows. The next section develops the prior--free model of finance and proves the fundamental theorem of asset pricing. We also provide a Harrison--Kreps--like theorem on the relation between no arbitrage and economic viability. Section 3 contains our treatment of the superhedging problem as a linear program in infinite dimensions.

\section{The Fundamental Theorem of Asset Pricing}

The usual approach to a no--arbitrage based theory of asset pricing starts with a probability space $\left(\Omega, \cF, P\right)$; future asset prices are  taken to be nonnegative random variables $S_d: \Omega \to \mathbb R_+$ that can be bought at prices $f_d \ge 0$ today.  In particular, one implicitly assumes that the distribution function
of the future asset prices is known.  A fortiori, it is presumed that the null sets, or prices that we can possibly never observe, are known. In light of the recent discussion about Knightian uncertainty and developments in finance concerning ambiguity about stochastic volatility\footnote{See \cite{Vorbrink10} and \cite{EpsteinJi11}, e.g.}, one might question this assumption.

  The fundamental theorem of asset pricing  states that under ($P$--almost sure) no arbitrage there exists an \emph{equivalent} martingale measure $P^*$; this measure has the same null sets as $P$ and satisfies the pricing or martingale relation
$$E^{P^*} S_d = f_d\qquad (d=0,\ldots,D)\,.$$

In discrete models, it is easy to develop the fundamental theorem and the subsequent hedging analysis without any probabilistic assumption. When teaching these results, I have always asked myself where one really needs the prior $P$. I present here an alternative approach that is, ex ante, probability--free. Instead, I work on a measurable space with a nice topological structure.

\subsection{The Topological Model of Finance}

Let $(\Omega,d)$ be a Polish space, i.e.  $d$ is a complete metric on $\Omega$ that allows for a countable dense subset $\mathbb Q :=\left\{(q_n) : n=1,2,\ldots,\right\} \subset \Omega$. Let $\cF$ be the Borel $\sigma$--field on $\Omega$, generated by the open subsets.

There are $D+1$ financial assets traded at time $0$ at prices $f_d \ge 0, d=0,1,\ldots, D$ with an uncertain payoff at time $1$.
As usual in Finance, we assume that there is a riskless asset $S_0(\omega)=1$ for all $\omega\in\Omega$ with price $f_0=1$ (so the interest rate is normalized to $0$). The uncertain assets have a payoff $S_d(\omega)$ that is continuous in $\omega$. Continuity in $\omega$ is stronger than mere measurability, of course; but it comes for free in most models. Note that we speak of continuity in the state variable, not in time, here. Usually, one can construct the model in such a way that $S_d$ is a projection mapping on some product space. These mappings are automatically continuous in the associated topology.

\begin{definition}
  A portfolio is a vector $\pi\in\mathbb R^{D+1}.$  A portfolio $\pi$ is called an arbitrage if we have
  \begin{align*}
    \pi \cdot f = \sum_{d=0}^D \pi_d f_d &\le 0 \\
    \pi \cdot S(\omega) &\ge 0 \qquad \mbox{for all $\omega\in\Omega$}\\
    \pi \cdot S(\omega) &> 0 \qquad \mbox{for some $\omega\in\Omega$}\,.
  \end{align*}
  We say that the market $(f,S)$ is arbitrage--free if there exists no arbitrage $\pi$.
\end{definition}

Note how we adapt the notion of arbitrage. Instead of assuming that a riskless gain is possible $P$--almost surely, we require a riskless gain for all states $\omega$, and a strictly positive gain for some state. By continuity of asset payoffs in states, this leads to a strictly positive payoff on an open subset $U$ of the state space.

\subsection{The Fundamental Theorem}

The classic fundamental theorem of asset pricing says that markets are arbitrage free if and only if one has a pricing or martingale measure $P^*$ that is equivalent to the prior $P$. The equivalence is important as the pricing measure must assign a positive price to all events that have positive probability under the prior $P$. In our version of the fundamental theorem, the equivalence to $P_0$ is replaced by the requirement of full support\footnote{The support of a probability measure $Q$ is the smallest closed set $F\subset \Omega$ with $Q(F)=1.$}: the pricing measure needs to assign positive prices to all open sets $U \subset \Omega$.

We fix our language with the following definition.

\begin{definition}
  A probability measure $P$ on $(\Omega,\cF)$ is called a martingale measure (for the market $(f,S)$) if it satisfies
  $\int S_d \, dP=f_d, d=1,\ldots,D$. If $P$ has also full support, we call it full support martingale measure.
\end{definition}

We are now able to characterize the absence of arbitrage in our model.

\begin{theorem}\label{TheoremFTAP}
  The market $(f,S)$ is arbitrage--free if and only if there exists a full support martingale measure. 
\end{theorem}

\begin{proof}(Easy Part)
Suppose that we have a full support martingale measure $P^*$. Let $\pi$ be an arbitrage. By continuity of $S$ in $\omega$,  there exists an open set $U\subset\Omega$ such that
$$\pi \cdot S(\omega) > 0$$ holds true for all $\omega \in U$. As $P^*$ has full support, it puts positive measure on $U$.  $\pi \cdot S$ is nonnegative and positive on a set of positive $P^*$--measure, so we have  $E^{P^*}  \pi \cdot  S = \int \pi \cdot S \, dP^* > 0$. The contradiction
$$0 \ge f \cdot \pi = \pi \cdot E^{P^*} S =E^{P^*} \left[ \pi \cdot  S\right] > 0$$ follows. This proves the easy part of the theorem.
\end{proof}

The more demanding direction of the fundamental theorem requires some preparation. First, let us note that measures with full support exist on Polish spaces. Just take our dense sequence $(q_n)_{n=1,2,\ldots}$ and define a probability measure $$\mu_0:=\sum_{n=1}^\infty 2^{-n} \delta_{q_n}$$ that puts weight $2^{-n}$ on $q_n$. As the sequence is dense, every nonempty open subset $U$ contains at least one $q_n$; hence $\mu_0(U)>0$.

We can now go on  and find even probability measures with full support under which all payoffs are integrable (note that we did not assume that the payoffs $S_d$ are bounded functions). Just let $M(\omega):=\max_{d=1,\ldots,D} S_d(\omega)$ and define a new measure $\mu_1$ via
$$\mu_1(A) = \int_A \frac{c}{1+M} d\mu_0$$ for Borel sets $A$ and the normalizing constant $c=\left( \int_\Omega \frac{1}{1+M} d\mu_0\right)^{-1}$.  As the density is strictly positive, $\mu_1$ has full support. Every $S_d$ is $\mu_1$--integrable because of $$\frac{S_d}{1+M} \le 1 \,.$$

We conclude that the  set
$C$ of probability measures on $(\Omega,\cF)$ with full support that make each $S_d$ integrable is non--empty. It is clear that $C$ is convex.

After this preparation, we can now begin with the typical separation argument (and we follow the exposition in \cite{FoellmerSchied02book} here). Let $R_d:=S_d-1, d=1,\ldots,D$ be the returns of the uncertain assets. Define the nonempty convex subset of $\mathbb R^D$
$$K:=\left\{ \begin{pmatrix}
               \int R_1 d\mu \\
               \vdots \\
               \int R_D d\mu \\
             \end{pmatrix}
: \mu\in C \right\}\,.$$ A full support martingale measure exists if and only if  $0\in K$. So let us assume, to achieve a contradiction, that $0\notin K$.

By the separation theorem in $\mathbb R^D$, there exists a vector $0\not= \phi \in \mathbb R^D$ such that $\phi \cdot k \ge 0$ for all $k\in K$ and $\phi \cdot k >0$ for some $k\in K$.
As we have $$\int_\Omega \phi \cdot R(\omega) \mu(d\omega) \ge 0 $$ for all $\mu\in C$, we must have $$\phi \cdot R(\omega) \ge 0$$ for all $\omega \in \Omega$.

The additional strict inequality
 $$\int_\Omega \phi \cdot R(\omega) \mu(d\omega) > 0 $$
for some $\mu \in C$ implies that there must exist $\omega_0 \in\Omega$ with
$\phi \cdot R(\omega_0) >0$.

Now define a portfolio $\pi$ by setting $\pi_0=-\sum_{d=1}^D \phi_d f_d $ and $\pi_d=\phi_d, d=1,\ldots,D$. Then $\pi \cdot f =0$, and
$$\pi \cdot S(\omega)= \pi_0 + \sum_{d=1}^D \pi_d S_d(\omega)
= \sum_{d=1}^D \pi_d Y_d(\omega) \ge 0\,,$$
as well as $\pi \cdot S(\omega_0)>0$. Hence, $\pi$ is an arbitrage, a contradiction.
We conclude that $ 0 \in K$.  This proves the fundamental theorem of asset pricing in our topological context.

\subsection{The Harrison--Kreps Theorem}

After the appearance of the Black--Scholes--Merton theory, a certain confusion about the magic of the formula and its supposed independence of any assumptions on preferences reigned.
Harrison and Kreps clarified the situation in their seminal paper by relating the concept of no arbitrage and economic viability. We  show how to obtain an analog to their result in our prior--free  model.

Let us assume in this section that $(\Omega,d)$ is also compact.

Let $\cX:=C(\Omega,d)$ be our space of net contingent trades at time $1$. Our financial market $(f,S)$ generates a marketed subspace
$$M:=\langle S_0,S_1,\ldots,S_D\rangle = \left\{ \pi \cdot S ; \pi \in \mathbb R^{D+1}\right\}\,.$$
Under no arbitrage, two portfolios $\pi$ and $\psi$ that lead to the same payoff $\pi \cdot S = \psi \cdot S$ must have the same value at time $0$, or $\pi \cdot f = \psi \cdot f$.
The linear mapping -- the price functional --
$$\phi : M \to \mathbb R$$ given by
$\phi\left(\pi\cdot S\right)=\pi \cdot f$
is then well--defined.

An economic agent is given by  a complete and transitive   relation $\succeq$ on $\cX$. In addition, we assume that $\succeq$ is continuous in the sense that for all $X\in\cX$ the better--than-- and worse--than--sets
$$ \left\{Z\in\cX : Z \succeq X\right\} ,  \left\{Z\in\cX :  X \succeq  Z \right\}
$$ are closed (in the norm topology on $\cX$). Finally, the agent's preferences are strictly monotone: if $X \in \cX$ satisfies $X \ge 0$ and $X\not=0$, then for all $Z\in\cX$, we have
$Z+X \succ Z$ (with the obvious definition of $\succ$).

We say that the market $(f,S)$ is viable if there exists an agent $\succeq$ such that no trade is optimal given the budget $0$: for all $\pi\in\mathbb R^{D+1}$ with $\pi \cdot f \le 0$ we have $0 \succeq \pi \cdot S $.

%
%

\begin{theorem}
  The market  $(f,S)$ is viable if and only if there exists a strictly positive linear functional $\Phi : \cX\to\mathbb R$ such that
   $\Phi(X)=\phi(X)$ for $X\in M$.
\end{theorem}

\begin{proof}
A viable market $(f,S)$ must not admit arbitrage opportunities, as the agent would improve upon any portfolio by adding the arbitrage.
  If $(f,S)$ is arbitrage--free, the fundamental theorem \ref{TheoremFTAP} allows us to pick a full--support martingale measure $\mu$. We set $\Phi(X)=\int_\Omega X d\mu$ and obtain our desired extension of $\phi$. The extension is strictly positive as $\mu$ has full support.

  On the other hand, if $\Phi$ is a strictly positive extension of $\phi$, the linear preference relation
  $$X \succeq Y \quad\mbox{if and only if}\quad \Phi(X)\ge\Phi(Y)$$ defines an agent for whom no trade is optimal given initial budget $0$. Note that positive linear functionals on the Banach lattice $C(\Omega,d)$ are continuous. Thus, the preference relation is also continuous.
\end{proof}

\section{(Super)Hedging of Contingent Claims}

In the absence of complete markets, the benchmark for imperfect hedging procedures are given by super-- resp.~subhedging a given claim where one aims to find the cheapest portfolio that stays above resp.~the most expensive portfolio that stays below the claim. We develop here a purely non--probabilistic approach to superhedging by embedding the problem fully into the language of (infinite--dimensional) linear programming. To the best of our knowledge, this has not been done before.

In our setting, superhedging a claim is a linear problem where the variable, the portfolio, is finite--dimensional, but where the constraint, staying on the safe side, is infinite--dimensional, in the space of continuous (bounded) functions. As the dual of this space consists of countably additive measures, the dual program has thus  infinite--dimensional variable space. We show that the constraints imply that the dual program is to maximize the expectation of the claim over all martingale measures.

It is noteworthy that the dual constraints do not require that we look at \emph{full support} martingale measures. From the point of view of optimization, it is more natural to look at the set of martingale measures only as this set is closed, and we can thus find an optimal solution to our problem.

From a Finance point of view, one wants to establish the result that the cheapest superhedge is equal to the least upper bound for all no arbitrage prices. No arbitrage prices, however, are determined by full support martingale measures. The two numbers coincide  because the set of all full support martingale measures is dense in the set of martingale measures. The maximizer of the dual problem is not attained by such a full support martingale measure; a superhedging portfolio would be an arbitrage if its value was a no arbitrage price. Hence, the value cannot be attained by a full support martingale measure.

\subsection{Claims and No Arbitrage Prices}

\begin{definition}
  A derivative or (contingent) claim is a continuous mapping $H:\Omega \to \mathbb R_+$.
  $h \ge 0$ is called a no arbitrage price for $H$ if the extended market with $D+2$ assets and $f_{D+1}=h$ and $S_{D+1}=H$ admits no arbitrage opportunities.
\end{definition}

As usual, we can now apply the fundamental theorem to obtain a characterization of no arbitrage prices. If we take the expected value of the claim under a full support martingale measure, the new, extended market does have such a full support martingale measure as well, and the expected value is a no arbitrage price. If, on the other hand, $h$ is a no arbitrage price, the other direction of the fundamental theorem implies that it is the expected value under some full support martingale measure for the extended market. This measure has to be a full support martingale measure for the old market as well.

\begin{corollary}
  $h$ is a no arbitrage price for a claim $H$ if and only if $$h=\int H \,dP$$ for a full support martingale measure $P$.
\end{corollary}

\subsection{Superhedging and Linear Programming}

We show now how to tackle the problem of superhedging in a probability--free way. We keep the assumption that $\left(\Omega,d\right)$ is compact.

\begin{definition}
  A portfolio $\pi$ is called a superhedge for the claim $H$ if $\pi \cdot S (\Omega) \ge H(\omega)$ holds true for all $\omega\in\Omega$.  A portfolio $\pi$ is called a subhedge for the claim $H$ if $\pi \cdot S (\Omega) \le H(\omega)$ holds true for all $\omega\in\Omega$.
\end{definition}

 A seller of the claim $H$ naturally  asks for the cheapest superhedge. This leads to the following linear optimization problem.

\begin{problem}[Problem SH]
Find the cheapest superhedge for the claim $H$; minimize $\pi \cdot f$ over $\pi \in \mathbb R^{D+1}$ subject to $\pi \cdot S (\Omega) \ge H(\omega)$ for all $\omega\in\Omega$.
\end{problem}

There is a natural subhedging companion to the above problem; it can be treated by the same methods as
the superhedging problem. We leave the details to the reader.

Let us now formulate the dual program for the superhedging problem SH.

It is usual to work with dual pairs of vector spaces  in linear programming. Here, we have the (finite--dimensional) space of portfolios $\mathbb R^{D+1}$ whose dual is the space  itself, of course. The bilinear form is the usual scalar product that we denote by  $x\cdot y$ for $x,y \in \mathbb R^{D+1}$.

The linear superhedging constraint can be described by the claim $H$ and the linear mapping
$$B: \mathbb R^{D+1} \to C(\Omega,d)$$ with $B\pi=\pi\cdot S$
which maps portfolios to their payoffs in the space of continuous functions.

The dual space of $C(\Omega,d)$ is the space of all (countably additive) finite measures  $ca(\Omega,\cF)$. The bilinear form on these spaces is given by
$$\bilinear{X}{\mu}=\int_\Omega X d\mu,\qquad\left(X\in C(\Omega,d), \mu \in ca(\Omega,\cF)\right)\,.$$

In order to formulate the dual program, we need the adjoint mapping to $B$.
In our context, the adjoint mapping  is $B^* : ca(\Omega,\cF) \to \mathbb R^{D+1} $ given by
$$B^* \mu = \int_\Omega S d\mu\,.$$ It maps measures to the candidate security prices (if that measure was taken as a pricing measure).

As one easily checks, $B^*$ satisfies the defining condition for an adjoint mapping
$$\bilinear{ B \pi}{\mu} = \bilinear{\pi}{B^* \mu} \qquad\left(\pi \in \mathbb R^{D+1}, \mu\in ca\left(\Omega,\cF\right)\right)\,.$$

As the superhedging problem has no sign constraints on the portfolio, the dual program has to have an equality constraint of the form $B^* \mu = f$. The inequality $B \pi \ge H$ leads to the nonnegativity constraint $\mu \ge 0$ in the dual program. In the dual program, we thus maximize over all positive measures $\mu$ that satisfy the integral constraints
$$\int S_d d\mu= f_d, d=0,\ldots, D\,.$$
As $S_0=f_0=1$, such a $\mu$ is a probability measure. The other constraints say that $\mu$ is a martingale measure.

\begin{problem}[DSH]
Minimize the prices $\int_\Omega H d\mu$ over all martingale measures $\mu \in ca(\Omega,\cF)$.
\end{problem}

We are now ready to state the main duality result on probability--free superhedging.

\begin{theorem}
\begin{enumerate}
  \item The linear programs SH and DSH are dual to each other. Both problems have the same value.
  \item  Both programs have optimal solutions; in particular, there exists a superhedge $\pi^*\in\mathbb R^{D+1}$ and a martingale measure $P^*$ such that
      $$\pi^*\cdot f = \int H dP^*\,.$$
\end{enumerate}

\end{theorem}

The details of the proof are postponed to the next subsection. Here, we want to point out that our (semi)--infinite--dimensional setup requires some care as the strong duality of the finite--dimensional linear programming theory carries over to infinite dimensions only under additional technical assumption. For our setup, it is sufficient to show that the cone
$$ C:= \left\{ \left( \int S d\mu, \int H d\mu\right) : \mu \in ca\left(\Omega,\cF\right)_+\right\}\subset \mathbb R^{D+2}$$ is closed. This follows from the fact that the cone of all positive measures in $ca\left(\Omega,\cF\right)$ is metrizable; a converging sequence in $C$ needs to stay bounded because of $S_0=1$. As such sequences are relatively weak*--compact, we find a converging subsequence of measures; see the next subsection for more details.

From a Finance point of view, one would like to go a step further and show that the superhedging value is also the least upper bound for all no arbitrage prices. But this is immediate from the fact that the set of full support martingale measures is dense in the set of martingale measures.
However, a maximizer cannot be  in the set of full support martingale measures as the superhedging portfolio leads to an arbitrage.

\begin{corollary}
\begin{enumerate}
\item  The price of the cheapest superhedge for a claim $H$ is equal to the least upper bound on all no arbitrage prices for $H$.
\item The price of the most expensive subhedge for a claim $H$ is equal to the greatest lower bound on all no arbitrage prices for $H$.
\end{enumerate}

\end{corollary}

\subsubsection{Details on Embedding the Superhedging Problem into the Language of Linear Programming in Infinite--Dimensional Spaces}

We now show how to embed the above superhedging problem and its dual  into the language of linear programming in infinite--dimensional spaces. We will then obtain a dual characterization of the value functions via the infinite--dimensional strong duality theorem. We use the language of the textbook by \cite{AndersonNash87}. Note, however, that we have swapped the primal and the dual program in our formulation.

Before we start our analysis, let us note that both problems are  consistent in the sense that portfolios satisfying the constraints exist: As there exists a riskless asset and $H$ is bounded, superhedges exist. By no arbitrage, there exist martingale measures.

We now embed our problems into the formulation of Anderson and Nash. Let $X=ca\left(\Omega,\cF\right)$
and $Y=C\left(\Omega,d\right)$. The (separating) bilinear form on $X$ and $Y$ is
$$\bilinear{\mu}{K}=\int_\Omega K(\omega) \mu(d\omega)$$ for  $\mu \in X$ and $K\in Y$.
We let $Z=W=\mathbb R^{D+1}$ with the usual bilinear form given by the scalar product. Let us set $c:=-H$. Recall that we introduced the linear mapping $B^* : X \to Y$ with $B\mu=\int S d\mu$ above.  Let $A=-B^*$. The primal program in the sense of Anderson and Nash is then

\begin{problem}[EP]
minimize $\bilinear{\mu}{c}$ over $\mu\in X$  subject to $A \mu = f$ and $\mu \ge 0$.
\end{problem}

Note that this is exactly our problem DSH above.

The adjoint mapping to $A$ is  $A^*=-B$, of course.
The dual problem is now
\begin{problem}[$\mbox{EP}^*$]
maximize $\bilinear{f}{\pi}$ over $\pi\in W$  subject to $A^* \pi \le c  $.
\end{problem}
This is our superhedging problem SH.

According to Theorem 3.10 in \cite{AndersonNash87}, strong duality holds true if the set
$$ C:= \left\{ \left( \int S d\mu, \int c d\mu\right) : \mu \in ca\left(\Omega,\cF\right)_+\right\}\subset \mathbb R^{D+2}$$
is closed. Let $(\mu_n)$ be a sequence of nonnegative finite measures such that
$$ \int S d\mu_n \to z \in \mathbb R^{D+1}, \int H d\mu_n \to r \in \mathbb R$$ as $n \to \infty$. As the zeroth asset is riskless, we conclude that
$\mu_n(\Omega)=\int S_0 d\mu_n$ remains bounded. In particular, the sequence $(\mu_n)$ is relatively compact in the weak*-- (or $\sigma\left(ca\left(\Omega,\cF\right),C\left(\Omega,d\right)\right)$--)topology. As this topology is metrizable (on the set of nonnegative measures, e.g. by the Prohorov distance), we can assume without loss of generality that $(\mu_n)$ converges in the weak*--topology to some $\mu \ge 0$. By continuity of the integral with respect to weak*--convergence, we then get
$$z=\int S \,d\mu, r = \int c \,d\mu$$ and this establishes that $C$ is closed.

We finally have to show that optimal solutions to both problems exist. The set of martingale measures is weak*--compact because it is a closed subset of the weak*--compact set of measures that are bounded by $1$ (Alaoglu's theorem). Hence, a maximizing martingale measure for the continuous linear evaluation $\int H d\mu$ exists.

The argument for the existence of an optimal superhedge is somewhat lengthier (but follows the same route of arguments as in the probabilistic case). Let $(\pi_n)$ be a minimizing sequence of portfolios. If the sequence is bounded, there exists a converging subsequence with limit $\pi^*$, and this is the desired optimal superhedge.

So let us assume that $(\pi_n)$ is unbounded.The sequence $$\eta_n:=\frac{\pi_n}{\|\pi_n\|}$$ is then well--defined for large $n$. Without loss of generality, it converges to some $\eta$ with norm $1$.

We can assume that the asset payoffs $S_0,\ldots, S_D$ are linearly independent, or else we choose a maximal linearly independent market including the riskless asset.
As $\pi_n$ are superhedges, and $H$ is bounded, we obtain
$$ \eta^* \cdot S = \lim \eta_n \cdot S \ge \limsup H / \|\pi_n\| = 0\,.$$
As $\pi_n \cdot f$ converges to the finite value of the superhedging problem, we also get
$$\eta^* \cdot f = 0\,.$$ By no arbitrage, we conclude $\eta^*=0$ which is a contradiction to $\|\eta^*\|=1$.

\section*{Conclusion}

A theory of hedging and pricing of derivative securities can be developed without referring to a prior probability measure that fixes the sets of measure zero. In this sense, one can say that the ideas of hedging and pricing (as far as it is related to and based on hedging) are independent from probability theory, from an epistemological point of view.

When one starts an economic model of uncertain markets,  a financial model, one can impose different strengths of probabilistic sophistication, either implicitly or explicitly. The strongest form of the efficient market hypothesis requires that there is an ``objective'' probability for all events, known by the market, and reflected by returns. The weaker form of the efficient market hypothesis, as nowadays commonly used in Mathematical Finance, just requires that the market participants share the same view of all null and probability one events. Returns and derivative prices are then determined by a martingale measure that is equivalent to the prior probability. Our approach goes a step further and relaxes the requirement that agents and the market agree on null events. It is then still possible to develop a reasonable theory of derivative pricing. The no arbitrage condition is strong enough to introduce a pricing probability even without any probabilistic prior assumption; this pricing probability has full support on the state space, and thus assigns positive probability to all open sets.

\ifx\undefined\BySame
\newcommand{\BySame}{\leavevmode\rule[.5ex]{3em}{.5pt}\ }
\fi
\ifx\undefined\textsc
\newcommand{\textsc}[1]{{\sc #1}}
\newcommand{\emph}[1]{{\em #1\/}}
\let\tmpsmall\small
\renewcommand{\small}{\tmpsmall\sc}
\fi

\end{document}